\documentclass[draftcls,onecolumn,12pt]{IEEEtran}

\usepackage [latin1]{inputenc}
\usepackage[usenames,dvipsnames]{xcolor}
\usepackage[tbtags]{amsmath}
\usepackage{amsthm}
\usepackage{subfigure}
\usepackage{graphicx}
\usepackage{amssymb}
\usepackage{caption}
\usepackage{threeparttable}

\newtheorem{theorem}{Theorem}

\newtheorem{lemma}{Lemma}

\newtheorem{definition}{Definition}

\usepackage[numbers,sort&compress]{natbib}

\begin{document}

\title{A New Metric and Its Scheme Construction for Evolving $2$-Threshold Secret Sharing Schemes}

\author{Wei Yan,
        Sian-Jheng Lin,~\IEEEmembership{Member,~IEEE}
        \thanks{Yan is with the School of Cyber Science and Technology, University of Science and Technology of China (USTC), China, e-mail: yan1993@mail.ustc.edu.cn.}
        \thanks{Lin is with the Theory Lab, Cenrtal Research Institute, 2012 Labs, Huawei Technology Co. Ltd, Hong Kong, China, email:lin.sian.jheng1@huawei.com}
}

\maketitle
\begin{abstract}
Evolving secret sharing schemes do not require prior knowledge of the number of parties $n$ and $n$ may be infinitely countable.	
It is known that the evolving $2$-threshold secret sharing scheme and prefix coding of integers have a one-to-one correspondence.
However, it is not known what prefix coding of integers to use to construct the scheme better.
In this paper, we propose a new metric $K_{\Sigma}$ for evolving $2$-threshold secret sharing schemes $\Sigma$.
We prove that the metric $K_{\Sigma}\geq 1.5$ and construct a new prefix coding of integers, termed $\lambda$ code, 
to achieve the metric $K_{\Lambda}=1.59375$.
Thus, it is proved that the range of the metric $K_{\Sigma}$ for the optimal 
$(2,\infty)$-threshold secret sharing scheme is $1.5\leq K_{\Sigma}\leq1.59375$.
In addition, the reachable lower bound of the sum of share sizes for $(2,n)$-threshold secret sharing schemes is proved.
\end{abstract}

\section{Introduction}
The secret sharing scheme was first proposed independently by Shamir~\cite{SS1} and Blakley~\cite{SS2} in 1979.
Its application scenario is that there is sensitive information that needs to be safely stored, 
so the sensitive information is first divided into $n$ shares, and each share is assigned to a party.
When assigning shares, some specific subsets of $n$ parties are set as qualified subsets. 
This sensitive information can only be recovered when the parties in the qualified subsets appear together. 
Conversely, when $m$ parties cannot form a qualified subset, 
they cannot recover the sensitive information, or even obtain any relevant content about the sensitive information. 
Secret sharing has been applied in widespread applications, such as verifiable signature sharing~\cite{VSS},
threshold digital signatures~\cite{TC1,TC2} and electronic voting~\cite{EV99}.

The secret sharing schemes first proposed by Shamir~\cite{SS1} and Blakley~\cite{SS2} are $(t,n)$-threshold secret sharing, which means that any $t$ parties among $n$ parties are a qualified subset.
The general secret sharing schemes were introduced by Ito \emph{et al.}~\cite{GSS}.
Traditional secret sharing schemes assume that the number of parties $n$ is known in advance 
or that the upper bound of $n$ is estimated.
However, this assumption comes at a cost: 
when $n$ is estimated to be too small, secret sharing needs to be re-made; 
when $n$ is estimated to be too large, it may cause waste.

Recently, the evolving secret sharing scheme was introduced by Komargodski \emph{et al.}~\cite{ESS0,ESS1}.
This new variant of secret sharing does not require prior knowledge of the upper bound of $n$ 
and $n$ may be infinitely countable. 
Komargodski \emph{et al.} found a one-to-one correspondence between the evolving $2$-threshold secret sharing scheme and  prefix coding of integers. 
In 2018, D'Arco \emph{et al.}~\cite{ESS2} reinterpreted the equivalence between the evolving $2$-threshold secret sharing scheme and prefix coding of integers with a new perspective. 
In 2020, Okamura and Koga~\cite{ESS3} extended the shared secret from $1$-bit to any $l$-bit based on the work of D'Arco \emph{et al.}, and by combining Shamir's secret sharing scheme, 
proposed an evolving 2-threshold secret sharing scheme using $D$-ary prefix codes.
In addition, there are studies on evolving secret sharing schemes for dynamic thresholds and robustness~\cite{RESS}, 
probabilistic evolving secret sharing~\cite{PESS} and evolving ramp secret sharing~\cite{ERSS1,ERSS2}.

Although the evolving 2-threshold secret sharing scheme can be completely characterized by the prefix coding of integers, 
we do not know which prefix code should be chosen to construct the secret sharing scheme better.
We need a metric for judging which evolving 2-threshold secret sharing scheme constructed by the prefix coding of integers works better.

In this paper, we investigate which prefix coding of integers should be chosen to construct 
the evolving 2-threshold secret sharing scheme better.
The metric for evolving 2-threshold secret sharing proposed in this paper draws on 
the metric minimum expansion factor $K_{\sigma}^{*}$~\cite{yan,yan2} for universal coding of integers~\cite{Elias75}. 
The contributions of this paper are enumerated as follows.
\begin{enumerate}
\item The reachable lower bound of the sum of share sizes for $(2,n)$-threshold secret sharing schemes is proved.
\item A new metric for evolving $2$-threshold secret sharing schemes is proposed.
\item Under the new metric, evolving $2$-threshold secret sharing schemes whose effect is close to the optimal scheme are constructed.
\end{enumerate}

The paper is structured as follows. 
Section \ref{pre} introduces some background knowledge.
Section \ref{bound} proves the reachable lower bound of the sum of share sizes for $(2,n)$-threshold secret sharing schemes.
Section \ref{metric} proposes a new metric for evolving $2$-threshold secret sharing schemes.
Section \ref{sec_new} constructs evolving $2$-threshold secret sharing schemes whose effect is close to the optimal scheme under the new metric.
Section \ref{sec_con} concludes this work.

\section{Preliminaries}\label{pre}
We establish some necessary notations. 
Let $\mathbb{N}$ be the set of positive integers.
Let $\mathbb{B}:=\{0,1\}$, and let $\mathbb{B}^{*}$ be a set consisting of all finite-length binary strings.
$\#S$ denotes the cardinality of the set $S$.
$|\alpha|$ denotes the length of string $\alpha$.
For a positive integer $n$, let $[n]:=\{1,2,\dots,n\}$.
\subsection{Secret Sharing Scheme}
Let $\mathcal{P}=\{P_{1},P_{2},\dots,P_{n}\}$ denote the set of participants, 
and let $2^\mathcal{P}$ be the power set of the set $\mathcal{P}$.
$\mathcal{M}\subseteq 2^\mathcal{P}$ is called \emph{monotone} 
if for any $S_1\in \mathcal{M}$, and $S_1\subseteq S_2$ it holds that $S_2\in \mathcal{M}$.
Before defining secret sharing scheme, we first define the access structure as follows.
\begin{definition} \label{def1}
$\mathcal{M}\subseteq 2^\mathcal{P}$ is said to be an access structure if $\mathcal{M}$ is non-empty and monotone.
Elements in $\mathcal{M}$ are called qualified, and elements not in $\mathcal{M}$ are called unqualified.
\end{definition}
\begin{definition}\label{def2}
Let $t$ and $n$ both be positive integers, and $1\leq t\leq n$. 
The $(t,n)$-threshold access structure $\mathcal{M}$ refers to the set containing only all elements in $2^\mathcal{P}$ of size at least $t$, i.e.
\[
\mathcal{M}=\{ A\in 2^{\mathcal{P}} \,|\,  \#A \geq t \}.
\]
\end{definition}
The definition of secret sharing scheme is given based on the access structure.
\begin{definition} \label{def3}
A secret sharing scheme	$\Sigma$ for an access structure $\mathcal{M}$ consists of a pair of probabilistic algorithms
$(\mathcal{S},\mathcal{R})$.
The sharing algorithm $\mathcal{S}$ generates $n$ shares $sh_{P_1}^{(s)},sh_{P_2}^{(s)},\dots,sh_{P_n}^{(s)}$
according to the secret $s\in S$ and the number of participants $n$.
The recovery algorithm $\mathcal{R}$ outputs a string according to the shares of the subset $A\in2^\mathcal{P}$.
The algorithm is required to satisfy:
\begin{itemize}
\item \textbf{Correctness:} For every qualified set $A\in\mathcal{M}$ and any secret $s\in S$, the recovery algorithm $\mathcal{R}$ can recover the secret $s$ with probability $1$, i.e.
	\[
		Pr[\mathcal{R}(A,\{sh_{j}^{(s)}\}_{j\in A})=s]=1.
	\]
\item  \textbf{Secrecy:} Each unqualified set $A\notin \mathcal{M}$ does not get any information about the secret $s$;
that is, for any two different secrets $s_1,s_2\in S$ and each unqualified set $A\notin \mathcal{M}$, it holds that
the distributions $(\{sh_{j}^{(s_1)} \}_{j\in A})$ and $(\{sh_{j}^{(s_2)}\}_{j\in A})$ are the same.
\end{itemize}
\end{definition}
From Definition~\ref{def3},
it can be reasonably associated that $(t,n)$-threshold secret sharing refers to the secret sharing scheme of $(t,n)$-threshold access structure. 
For a secret sharing scheme, it is hoped to generate the sum of share sizes $\sum_{i=1}^{n}|sh_{P_i}^{(s)}|$ is as small as possible, which can make the amount of communication as small as possible.

The following introduces two important conclusions that need to be used.
\begin{theorem}~\cite{UR,ESS4} \label{thm2}
	Suppose that $\Sigma$ is $(t,n)$-thresholded secret sharing scheme for $1$-bit secret, and the $j$-th share size is $m_{j}$ bits, where $2\leq t\leq n$, $j\in[n]$ and $m_j\in \mathbb{N}$. Then, the sum of share sizes 
	\[
	\sum_{j=1}^{n}m_{j} \geq n\log_{2}(n-t+2).
	\]
	In particular, when $t=2$, the sum of share sizes for $(2,n)$-threshold secret sharing schemes
	$$\sum_{j=1}^{n}m_{j}\geq n\log_{2}n.$$
\end{theorem}
\begin{lemma}~\cite{UR,ESS4,ESS1}\label{lemma1}
	Suppose that $\Sigma$ is $(2,n)$-thresholded secret sharing scheme for $1$-bit secret, and the $j$-th share size is $m_{j}$ bits, where $j\in[n]$ and $m_j\in \mathbb{N}$. Then,
	\[
	\sum_{j=1}^{n}\frac{1}{2^{m_{j}}}\leq 1 .
	\]
\end{lemma} 
\subsection{Evolving Secret Sharing Scheme} 
Because the number of participants $n$ is uncertain 
and the upper bound of $n$ cannot be estimated in real scenarios, 
a class of secret sharing schemes needs to be defined so that $n$ can be infinitely countable.
Naturally, in this scenario, the parties participating in secret sharing will not be present at the same time. 
We assume that at the $t$-th moment, the $t$-th person arrives at the scene and asks for the distribution of subsequent shares, the previously distributed shares do not need to be changed, which can most effectively reduce the amount of communication.
Let $\mathcal{PN}=\{P_{1},P_{2},\dots,P_{n},\dots\}$ denote the set of participants.
Definitions~\ref{def1} and~\ref{def2} are naturally extended to the following definitions.
\begin{definition} ~\cite{ESS1}
Suppose that $\mathcal{M}\subseteq 2^{PN}$ is monotone, and for each time $t\in\mathbb{N}$,
$\mathcal{M}_{t}:=\mathcal{M}\cap \{P_{1},P_{2},\dots,P_{t}\}$ is an access structure, then $\mathcal{M}$ is said to be an evolving access structure.
\end{definition}
\begin{definition}~\cite{ESS1}
Let $t$ be a positive integer.
The evolving $t$-threshold access structure $\mathcal{M}$ refers to the set consisting only of all elements in $2^{\mathcal{PN}}$ of size at least $t$, i.e.
	\[
	\mathcal{M}=\{ A\in 2^{\mathcal{PN}}  \,| \, \#A \geq t \}.
	\]
\end{definition}
For simplicity, we use $(t,\infty)$-threshold to represent evolving $t$-threshold.
Now, the formal definition of the evolving secret sharing scheme is as follows.
\begin{definition} \label{def6}~\cite{ESS0,ESS1}
Let $S$ denote a domain of secrets, where $\#S\geq 2$.
Let $\mathcal{M}$ denote an evolving access structure.
An evolving secret sharing scheme $\Sigma$ for $S$ and $\mathcal{M}$ consists of a pair of probabilistic algorithms $(\mathcal{S},\mathcal{R})$.
The sharing algorithm $\mathcal{S}$ and the recovery algorithm $\mathcal{R}$ are required to satisfy:
\begin{itemize}
 \item[1)] At time $t\in \mathbb{N}$, the sharing algorithm $\mathcal{S}$ generates share $sh_{P_t}^{(s)}$ according to the secret $s\in S$ and shares $sh_{P_1}^{(s)},sh_{P_2}^{(s)},\dots,sh_{P_{t-1}}^{(s)}$, i.e.
	\[
		\mathcal{S}(s,\{sh_{P_i}^{(s)}\}_{i\in[t-1]})=sh_{P_t}^{(s)}.
	\]
\item[2)] \textbf{Correctness:} For each time $t\in \mathbb{N}$, every qualified set $A\in \mathcal{M}_t$ and any secret $s\in S$, the recovery algorithm $\mathcal{R}$ can recover the secret $s$ with probability $1$, i.e.
	\[
		Pr[\mathcal{R}(A,\{sh_{i}^{(s)}\}_{i\in A})=s]=1 .
	\]
\item[3)] \textbf{Secrecy:} For each time $t\in \mathbb{N}$, each unqualified set $A\notin \mathcal{M}_t$ does not get any information about the secret $s$;
that is, for each time $t$, any two different secrets $s_1,s_2\in S$ and each unqualified set $A\notin \mathcal{M}_t$, it holds that the distributions $(\{sh_{j}^{(s_1)} \}_{j\in A})$ and $(\{sh_{j}^{(s_2)}\}_{j\in A})$ are the same.
\end{itemize}
\end{definition}
Komargodski \emph{et al.}~\cite{ESS0,ESS1} found the equivalence between the $(2,\infty)$-thresholded secret sharing scheme 
and prefix coding of integers. 
This interesting conclusion is shown below.
\begin{theorem}~\cite{ESS0,ESS1}\label{thm1}
Let $\sigma: \mathbb{N}\rightarrow \mathbb{B}^{*}$ be a prefix coding of integers. 
The length of its $t$-th codeword is $L_\sigma(t)$, where $t\in\mathbb{N}$.
Such an integer code exists if and only if there exists a $(2,\infty)$-threshold secret sharing scheme $\Sigma$ for $1$-bit secret, and the $t$-th share size is $L_\sigma(t)$ bits.
\end{theorem}
An interesting and essential understanding of Theorem~\ref{thm1} can be found in~\cite{ESS2}. 

\subsection{The Metric Minimum Expansion Factor $K_{\mathcal{C}}^{*}$ for Universal Coding of Integers} 
Universal coding of integers is a class of binary prefix code,
such that the ratio of the expected codeword length to $\max\{1,H(P)\}$ 
is within a constant for any decreasing probability distribution $P$ of $\mathbb{N}$ (i.e., $\sum_{n=1}^{\infty}P(n)=1$, and $P(m)\geq P(m+1)\geq 0$ for all $m\in \mathbb{N}$), where $H(P):=-\sum_{n=1}^{\infty}P(n)\log_{2}P(n)$ is the entropy of $P$.
The formal definition of universal coding of integers is as follows.
\begin{definition} ~\cite{Elias75,yan}
Let $\sigma: \mathbb{N}\rightarrow \mathbb{B}^{*}$ be a binary prefix coding of integers.
Let $L_{\sigma}(\cdot)$ denote the length function of $\sigma$ so that $L_{\sigma}(m)=|\sigma(m)|$ for all $m\in \mathbb{N}$. 
$\sigma$ is called universal if there exists a constant $K_{\sigma}$ independent of $P$, such that
\begin{equation}\label{eq2}
	\frac{E_{P}(L_{\sigma})}{\max\{1,H(P)\}}\leq K_{\sigma},
\end{equation}
for any decreasing probability distribution $P$ with finite entropy, where $$E_{P}(L_{\sigma}):=\sum_{n=1}^{\infty}P(n)L_{\sigma}(n)$$ 
denotes the expected codeword length for $\sigma$.
Then $K_{\sigma}$ is called the expansion factor.
Let $K_{\sigma}^{*}:=\inf\{K_{\sigma}\}$ be the infimum of the set of expansion factors.
$K_{\sigma}^{*}$ is called the minimum expansion factor.
\end{definition}
The minimum expansion factor $K_{\sigma}^{*}$ is the smallest of the expansion factors.
For any universal coding of integers, its minimum expansion factor $K_{\sigma}^{*}$ is unique.
The minimum expansion factor $K_{\sigma}^{*}$, as a metric, evaluate the compression effect of universal coding of integers. 
Therefore, universal coding of integers $\sigma$ is called \emph{optimal} if $\sigma$ achieves the smallest $K_{\sigma}^{*}$.

Finally, we introduce two classes of universal coding of integers, $\iota$ code~\cite{yan2} and $\eta$ code~\cite{yan}, which can achieve smaller expansion factors. 
In particular, $\iota$ code is currently the only universal coding of integers that can achieve $K_{\iota}=2.5$.
The two codes are briefly introduced as follows.

The unary code $\alpha$ of the non-negative integer $m$ is constructed as $m$ bits of $0$ followed by a single $1$.
Let $\beta(m)$ denote the standard binary representation of $m\in\mathbb{N}$, 
and let $[\beta(m)]$ denote the removal of the most significant bit $1$ of $\beta(m)$.
For example, $\alpha(2)=001$, $\beta(10)=1010$ and $[\beta(10)]=010$.
Note that $[\beta(1)]$ is a null string.
We define $\beta(0)$ as a null string, and the length of the null string is $0$.
\begin{enumerate}
	\item The code $\iota: \mathbb{N}\rightarrow \mathbb{B}^{*}$ can be expressed as
\begin{equation*}
	\iota(m)=\left\{\begin{array}{lll}
		 1,                                                         &\text{if } m=1\text{,}   \\
		\alpha(\frac{|\beta(m)|}{2})0[\beta(m)],                     &\text{if } |\beta(m)| \text{ is even,} \\
		\alpha(\frac{|\beta(m)|-1}{2})\beta(m),                &\text{otherwise,}   \\
	\end{array}\right.
\end{equation*}
	for all $m\in \mathbb{N}$. The codeword length is given by
	\begin{equation*}
	\begin{aligned}
		|\iota(m)|&=1+\lfloor \frac{|\beta(m)|}{2} \rfloor+|\beta(m)|    \\
		&=2+\lfloor \frac{1+\lfloor\log_2(m)\rfloor}{2} \rfloor+\lfloor\log_2(m)\rfloor,         \\
	\end{aligned}
\end{equation*}
for $2\leq m\in \mathbb{N}$ and $|\eta(1)|=1$.
	\item The code $\eta: \mathbb{N}\rightarrow \mathbb{B}^{*}$ can be expressed as
\begin{equation*}
	\eta(m)=\left\{\begin{array}{ll}
		\alpha(\frac{|\beta(m-1)|}{2})\beta(m-1),                     &\text{if } |\beta(m-1)| \text{ is even,} \\
		\alpha(\frac{1+|\beta(m-1)|}{2})0[\beta(m-1)],                &\text{otherwise,}   \\
	\end{array}\right.
\end{equation*}
	for all $m\in \mathbb{N}$. The codeword length is given by
	\begin{equation*}
		\begin{aligned}
			|\eta(m)|&=1+\lfloor \frac{1+|\beta(m-1)|}{2} \rfloor+|\beta(m-1)|    \\
			&=3+\lfloor \frac{\lfloor\log_2(m-1)\rfloor}{2} \rfloor+\lfloor\log_2(m-1)\rfloor,
		\end{aligned}
	\end{equation*}
	for $2\leq m\in \mathbb{N}$ and $|\eta(1)|=1$.
\end{enumerate}
\section{The Reachable Lower Bound of the Sum of Share Sizes for $(2,n)$-Threshold Secret Sharing Schemes}\label{bound}
In this section, the reachable lower bound of the sum of share sizes for $(2,n)$-threshold secret sharing schemes is proved.
Let $\min[sum(2,n)]$ denote the minimum sum of share sizes in all $(2,n)$-threshold secret sharing schemes. 
The \emph{optimal $(2,n)$-threshold secret sharing scheme} is defined as the secret sharing scheme 
with the sum of share sizes $\min[sum(2,n)]$.

In the early 1990s, Kilian and Nisan first proposed and proved Theorem~\ref{thm2} in an email, 
but it was not officially published. 
This unpublished result has been mentioned in numerous papers 
and the proof by Kilian and Nisan was first published publicly in~\cite{ESS4}.
Theorem~\ref{thm2} shows that the sum of share sizes for $(2,n)$-threshold secret sharing schemes 
is greater than or equal to $n\log_{2}n$, i.e.
$$\min[sum(2,n)]\geq n\log_{2}n.$$

This section gives a tighter bound than the lower bound $n\log_{2}n$ and this bound is achievable; 
that is, this section gives the exact expression for $\min[sum(2,n)]$.
First, a correspondence similar to Theorem~\ref{thm1} is given.
\begin{theorem}\label{thm3}
Let $\sigma: [n]\rightarrow \mathbb{B}^{*}$ be a prefix code. 
The length of its $t$-th codeword is $L_\sigma(t)$, where $t\in[n]$.
Such a prefix code exists if and only if there exists a $(2,n)$-threshold secret sharing scheme $\Sigma$ for $1$-bit secret, and the $t$-th share size is $L_\sigma(t)$ bits.
\end{theorem}
\begin{proof}
\begin{itemize}
\item[(1)] \textbf{Sufficient:}
Suppose that there exists a $(2,n)$-threshold secret sharing scheme $\Sigma$ for $1$-bit secret, and the $t$-th share size is $L_\sigma(t)$ bits.
From Lemma~\ref{lemma1}, we obtain
\[
\sum_{j=1}^{n}\frac{1}{2^{L_\sigma(j)}}\leq 1 .
\]
Due to the Kraft's inequality~\cite{kraft}, we know that there exists a prefix code with codeword lengths $L_\sigma(1),L_\sigma(2),\dots,L_\sigma(n)$.
\item[(2)]\textbf{Necessary:}
Suppose that there is a prefix code with codeword lengths $L_\sigma(1),L_\sigma(2),\dots,L_\sigma(n)$.
Next, we construct a $(2,n)$-threshold secret sharing scheme $\Sigma$ for $1$-bit secret.
Let $s\in\mathbb{B}$ be the secret, and let $M$ denote the maximum value among $L_\sigma(1),L_\sigma(2),\dots,L_\sigma(n)$.
The dealer randomly generates a binary string $Q$ of length $M$.
Let $Q|_{t}$ denote the first $L_\sigma(t)$ bits of the string $Q$.
If $s=0$, then the $t$-th share is $sh(t)=\sigma(t)\oplus Q|_{t}$;
If $s=1$, then the $t$-th share is $sh(t)= Q|_{t}$.
The $t$-th share size is $L_\sigma(t)$ bits. 

In the recovery secret stage, let the two different shares be $sh(t_{1})$ and $sh(t_{2})$.
Without loss of generality assume that $|sh(t_1)|\leq|sh(t_2)|$.
If $sh(t_{1})$ is a prefix of $sh(t_{2})$, the output is $1$; otherwise, the output is $0$.
\begin{itemize}
\item[(2.1)]\textbf{Correctness:} If $s=0$, then $sh(t_1)=\sigma(t_1)\oplus Q|_{t_1}$ and 
$sh(t_2)=\sigma(t_2)\oplus Q|_{t_2}$. Since $Q|_{t_1}$ is a prefix of $Q|_{t_2}$ and $\sigma(t_1)$ is not a prefix of $\sigma(t_2)$, then $\Pi(t_1)$ is not a prefix of $\Pi(t_2)$. 
Therefore, the output is $0$ is correct.
If $s=1$, then $sh(t_1)=Q|_{t_1}$ and $sh(t_2)=Q|_{t_2}$.
Since $Q|_{t_1}$ is a prefix of $Q|_{t_2}$, the output is $1$ is correct.
\item[(2.2)]\textbf{Secrecy:}
Because $Q|_{t}$ is uniformly distributed on $\mathbb{B}^{L_\sigma(t)}$, whether $s=0$ or $s=1$, 
there is a share $sh(t)$ uniformly distributed on $\mathbb{B}^{L_\sigma(t)}$. 
Therefore, no single party can get any information about the secret $s$.
\end{itemize}
\end{itemize}
\end{proof}

Next, the main theorem of this section is given.
\begin{theorem} \label{thm4}
Let $n$ be an integer greater than $1$, then
\[
\min[sum(2,n)]=nm+2l,
\]
where $m:=\lfloor \log_{2}n \rfloor$ and $l:=n-2^{m}$.
\end{theorem}
\begin{proof}
Suppose that $\Sigma$ is a $(2,n)$-threshold secret sharing scheme for $1$-bit secret, and the $j$-th share size is $m_j$ bits, where $j\in[n]$.
We need to find a scheme $\Sigma$ that minimizes the sum $\sum_{j=1}^{n}m_{j}$; 
that is, find the minimum value $\min[sum(2,n)]$ of $\sum_{j=1}^{n}m_{j}$.

Due to Theorem~\ref{thm3}, $\Sigma$ corresponds to a prefix code $\sigma: [n]\rightarrow \mathbb{B}^{*}$,
and the length of its $j$-th codeword is $m_j$, where $j\in[n]$.
To make $\sum_{j=1}^{n}m_{j}$ minimum is equivalent to minimizing the
expected codeword length $L=\sum_{j=1}^{n}P_{6}(j)m_{j} =\frac{1}{n}\sum_{j=1}^{n}m_{j}$
of the prefix code $\sigma$ with probability distribution 
\[
P=\left(P(1)=P(2)=\dots=P(n)=\frac{1}{n}\right).
\]
Because given a probability distribution, Huffman code is the prefix code with the smallest expected codeword length.
Therefore, it is only necessary to calculate the expected codeword length 
when encoding with Huffman code in the case of the probability distribution $P$.
From $P$ is a uniform distribution and the encoding rule of Huffman code, 
it can be known that the code tree obtained by encoding is a full binary tree 
and the layers of leaf nodes differ by at most $1$, so the lengths of all codewords differ by at most $1$.
Then, there are $2^{m}-l$ codewords with codeword length $m$ and $2l$ codewords with codeword length of $m+1$.
Therefore, the minimum value of $\sum_{j=1}^{n}m_{j}$ is
\begin{equation*}
	\begin{aligned}
		\min[sum(2,n)]&= m(2^{m}-l)+2l(m+1)    \\
		&= m2^m+ml+2l    \\
		&= nm+2l.
	\end{aligned}
\end{equation*}
\end{proof}
From Theorem~\ref{thm4}, $nm+2l$ is the reachable lower bound of the sum of share sizes 
for $(2,n)$-threshold secret sharing scheme. 
Since $n\log_{2}n$ is the lower bound, there must be $nm+2l \geq  n\log_{2}n$. 
To show that our results are novel, 
we compare $nm+2l$ and $n\log_{2}n$ only from a purely mathematical point of view. 
\begin{lemma} \label{lemma2}
For any $n=2^{m}+l\in \mathbb{N}$, where $m=\lfloor \log_{2}n \rfloor$. Then
\begin{equation} \label{eq1}
	nm+2l \geq n\log_{2}n .
\end{equation}	
\end{lemma}
\begin{proof}
Let $x:=\log_{2}n$ and $y:=x-\lfloor x \rfloor$, then $0\leq y <1$.
We obtain
\begin{equation*}
	\begin{aligned}
		nm+2l \geq n\log_{2}n 	\iff & n\lfloor \log_{2}n \rfloor +2n-2 \cdot 2^{\lfloor \log_{2}n \rfloor} \geq n\log_{2}n    \\
		\iff & n(2+\lfloor \log_{2}n \rfloor-\log_{2}n ) \geq 2 \cdot 2^{\lfloor \log_{2}n \rfloor}   \\
		\iff & 2^{x}(2+\lfloor x \rfloor-x) \geq 2 \cdot 2^{\lfloor x \rfloor}       \\
		\iff & 2^{y}(2-y) \geq 2 .  \\
	\end{aligned}
\end{equation*}
Therefore, it is equivalent to proving that $2^{y}(2-y)\geq 2$ for $y\in [0,1)$.
Let $f(y):=2^{y}(2-y)$.
By calculating the derivative, we know that $f(y)$ is strictly monotonically increasing over the interval $[0,y_{0})$
and strictly monotonically decreasing over the interval $[y_{0},1)$, where $y_{0}=2-\dfrac{1}{\ln2}$.
Therefore, for $y\in [0,1)$, we have
\[
f(y)\geq \min\{f(0),f(1)\}=2.
\]
\end{proof}

From the proof process of Lemma~\ref{lemma2}, we know that when $y=0$, inequality~\eqref{eq1} is equal; 
that is, inequality~\eqref{eq1} is equal when $n$ is a power of 2. 
Due to the monotonicity of the function $f(y)$, 
we reasonably guess that at the midpoint $2^m+2^{m-1}$ of the two powers of 2, 
the difference between the two sides of inequality~\eqref{eq1} is large. 
Thus, Table~\ref{tab1} lists some values comparing $nm+2l$ and $n\log_{2}n$. 
Table~\ref{tab1} confirms the equivalence of inequality~\eqref{eq1} from the side; that is, when $n=2^m$, 
there is $nm+2l=m\cdot2^m=n\log_{2}n$. 
In addition, it can be found that the difference at $2^m+2^{m-1}$ is larger when $m$ is larger, interestingly, the latter difference is $2$ times the previous difference. 
This finding is verified below.
\begin{table}[t]
	\caption{Comparison of some values of $nm+2l$ and $n\log_{2}n$}\label{tab1}
	\centering
	\begin{threeparttable}
	\begin{tabular}{|c|c|c|c||c|c|c|c|}
		\hline
		$n$  &   $nm+2l$   &   $n\log_{2}n$  & difference  &  $n$  &   $nm+2l$   &   $n\log_{2}n$  & difference  \\
		\hline
		$2$   &    2  &   2       &  0     &  	$3$   &     5    &     4.75      &  0.25      \\
		$4$   &    8  &   8      &  0   &   	$6$   &     16     &   15.51     &    0.49  \\
		$8$   &     24  &   24   &  0   &     $12$   &      44      &   43.02      &   0.98  \\
		$16$   &    64  &   64   &  0  &   	  $24$   &      112      &   110.04    &   1.96    \\
		$32$   &   160  &   160  &  0   &  	  $48$   &      272     &   268.08   &   3.92   \\
		$64$   &   384  &   384   &  0   &     $96$   &     640 &    632.16  &  7.84  \\
		$2^7$   &  896  &   896   &  0   &    	$2^7+2^6$   &     1472  & 1456.31   & 15.69     \\
		$2^8$   &  2048  &  2048  &  0  &     	$2^8+2^7$   &  3328   &   3296.63 &   31.37 \\
		$2^9$   &  4608  &  4608  &  0  &      	$2^9+2^8$   &   7424  &   7361.25 &   62.75    \\
		$2^{10}$  &  10240  & 10240  &  0  &    $2^{10}+2^9$  &   16384  &   16258.50  &  125.50  \\
		$2^{11}$  &  22528  & 22528  &  0  &    $2^{11}+2^{10}$  &   35840  &  35589.00  & 251.00   \\
		$2^{12}$  &  49152  & 49152  &  0 &    	$2^{12}+2^{11}$  &   77824 &   77322.01  &  501.99   \\
		$2^{13}$  &  106496 &   106496   &  0 & $2^{13}+2^{12}$  &   167936  &  166932.02  & 1003.98 \\
		$2^{14}$  &  229376 &   229376  &  0 &  $2^{14}+2^{13}$  &  360448  &  358440.04 &   2007.96  \\
		$2^{15}$  &  491520 &   491520  &  0 &  $2^{15}+2^{14}$  &   770048  &  766032.08  &  4015.92    \\
		$2^{16}$  &  1048576 &  1048576 &  0 &	$2^{16}+2^{15}$  &   1638400  & 1630368.15   & 8031.85  \\
		\hline
	\end{tabular}
\begin{tablenotes}
\footnotesize 
\item  Note: After rounding $n\log_{2}n$, retain two decimal places.
\end{tablenotes}
\end{threeparttable}
\end{table}

Let $n_1:=2^m+2^{m-1}$ and $n_2:=2^{m-1}+2^{m-2}$.
Let $d(n):=nm+2l-n\log_{2}n$, we obtain
\begin{equation*}
	\begin{aligned}
		d(n_1)& =(2^m+2^{m-1})m+2^{m}-(2^m+2^{m-1})\log_{2}(2^m+2^{m-1})  \\
		&= (2^m+2^{m-1})(m-1)+2^{m}-(2^m+2^{m-1})\log_{2}(2^{m-1}+2^{m-2})         \\
		&= 2\Big[(2^{m-1}+2^{m-2})(m-1)+2^{m-1}-(2^{m-1}+2^{m-2})\log_{2}(2^{m-1}+2^{m-2}) \Big] \\
		&= 2d(n_2). \\
	\end{aligned}
\end{equation*}
This shows that when $m$ tends to infinity, the difference at $2^m+2^{m-1}$ tends to infinity.
Therefore, the reachable lower bound $nm+2l$ proved in this paper is not only a new result, but also very meaningful.

\section{A New Metric for $(2,\infty)$-Threshold Secret Sharing Schemes}\label{metric}
Theorem~\ref{thm1} shows that there is a one-to-one correspondence between the $(2,\infty)$-thresholded secret sharing scheme 
and the prefix coding of integers.  
For simplicity, integer codes mentioned in this paper refer to the prefix coding of integers.

The question that this section will aim to answer is which integer codes should be chosen to construct secret sharing better;
that is, this section strives to formulate a metric 
under which one can judge which of the two integer codes is more suitable for constructing $(2,\infty)$-thresholded secret sharing schemes.

In traditional secret sharing, since shares are distributed at one time, the metric to measure is the sum of share sizes, and it is hoped that the smaller the sum, the better. 
In evolving secret sharing, the specific number of parties participating in secret sharing is not known, and even the number of parties may be infinitely countable. 
Therefore, it is impossible to calculate the sum of share sizes by distributing the shares all at one time.

The simplest idea of whether it is possible to construct a $(2,\infty)$-threshold secret sharing scheme so that the length of share generated at any $t$ time is the smallest. Unfortunately, such a scheme does not exist. 
From Theorem~\ref{thm1}, for any $(2,\infty)$-thresholded secret sharing scheme $\Sigma$, 
there exists an integer code $\sigma$, such that the $t$-th codeword length $L_\sigma(t)$ is exactly 
the size of share distributed by the scheme $\Sigma$ at the $t$-th moment. 
From the theory of prefix codes, it is impossible to have an integer code $\sigma$, 
such that for any integer code $\psi$ and any positive integer $t$, 
there is $L_\sigma(t)\leq L_\psi(t)$. 
An integer code is called \emph{complete} if the integer code makes Kraft's inequality~\cite{kraft} equal.
When two different complete integer codes $\psi$ and $\sigma$ compare the codeword lengths, 
if $t$ is sufficiently large, both have $L_\sigma(t)\leq L_\psi(t)$, then there must be a smaller positive integer $t_0$ such that $L_\sigma(t_0) > L_\psi(t_0)$; in other words, 
the codeword length advantage of the $\sigma$ at larger integers is at the expense of codeword length at smaller integers. 
So this simple idea does not work.
 
The second common idea only considers sufficiently large moments $t$, if there is a $(2,\infty)$-threshold secret sharing scheme such that the share size $L(t)$ generated at sufficiently large moments $t$ are all the smallest, we think this scheme is the best.
It is known that the codeword length advantage at larger integers is at the expense of codeword length at smaller integers. However, in evolving secret sharing, the share size $L(t)$ of the larger moment $t$ is obviously not as important as the share size of the smaller moment due to the unknown number of parties. Therefore, this idea is unreasonable, and we need to find more reasonable and feasible metrics.

In fact, evolving secret sharing has one thing in common with universal coding of integers: they both face unknowns.
Evolving secret sharing has no prior knowledge of the number of parties, and universal coding of integers has no prior knowledge of probability distributions.
Therefore, it is hoped to propose a metric similar to the metric minimum expansion factor $K_{\sigma}^{*}$ to evaluate the overall effect of $(2,\infty)$-threshold secret sharing schemes.
The new metric is defined as follows.
\begin{definition} 
Let $\Sigma$ be a $(2,\infty)$-threshold secret sharing scheme, which corresponds to the integer code $\sigma$,
and the size of the share distributed at $t$-th moment is $L_\sigma(t)$ bits. 
The global metric $K_{\Sigma}$ of the scheme $\Sigma$ is defined as follows:
	\begin{equation}\label{eq3}
		\begin{aligned}
		K_{\Sigma}:=& \sup \left\{\dfrac{\sum_{t=1}^{n}L_\sigma(t)}{\min[sum(2,n)]}\,\Bigg|\, \forall n\in \mathbb{N}^{+},n\neq 1 \right\}     \\
		= &\sup \left\{\dfrac{\sum_{t=1}^{n}L_\sigma(t)}{	nm+2l }\,\Bigg|\, \forall n\in \mathbb{N}^{+},n\neq 1 \right\},
		\end{aligned}
	\end{equation}
where $m:=\lfloor \log_{2}n \rfloor$ and $l:=n-2^m$.   
\end{definition}

The meaning of the global metric $K_{\Sigma}$ is that no matter how many parties participate in secret sharing
when any fixed number of parties is $n_0$,
the sum of the share sizes for the $(2,\infty)$-threshold secret sharing scheme $\Sigma$ is less than or equal to
$K_{\Sigma}$ times of the sum of the share sizes for the optimal $(2,n_0)$-threshold secret sharing scheme.

Finally, the optimal $(2,\infty)$-threshold secret sharing scheme and the optimal integer code 
for $(2,\infty)$-threshold secret sharing schemes are defined below.

\begin{definition}
The $(2,\infty)$-threshold secret sharing scheme with the smallest global metric $K_{\Sigma}$ is called the optimal $(2,\infty)$-threshold secret sharing scheme. 
The integer code corresponding to the optimal $(2,\infty)$-threshold secret sharing scheme is called the optimal integer code 
for $(2,\infty)$-threshold secret sharing schemes.	
\end{definition}

\section{The Value Range of $K_{\Sigma}$ for the Optimal $(2,\infty)$-Threshold Secret Sharing Scheme}\label{sec_new}
In this section, we study the value range of the global metric $K_{\Sigma}$ for the optimal $(2,\infty)$-threshold secret sharing scheme. 

First, the lower bound of the global metric $K_{\Sigma}$ is given.
For simplicity, let
\[
\mathcal{L}(n,\sigma):=\dfrac{\sum_{t=1}^{n}L_\sigma(t)}{nm+2l}.
\]
Consider the case where the number of parties $n=2$.
Due to $L_{\sigma}(1)\geq 1$ and $L_{\sigma}(2)\geq 2$ for any integer code $\sigma$, we obtain
\[
\mathcal{L}(2,\sigma)=\dfrac{\sum_{t=1}^{2}L_\sigma(t)}{ 2\times1+2\times0 } \geq \dfrac{3}{2}.
\]
Therefore, the global metric $K_{\Sigma}$ satisfies $K_{\Sigma}\geq 1.5$.

Next, we will construct schemes $\Sigma$ in the following two subsections so that its the global metric $K_{\Sigma}$ is close to the lower bound $1.5$.

\subsection{Panning Code}
In this subsection, we construct $(2,\infty)$-threshold secret sharing schemes 
with small $K_{\Sigma}$ using known universal coding of integers.

The length of the second codeword of the currently constructed universal coding of integers 
$\sigma$ is strictly greater than $2$; that is, $L_{\sigma}(2)\geq 3$, then 
\[
\mathcal{L}(2,\sigma)=\dfrac{\sum_{t=1}^{2}L_\sigma(t)}{2} \geq 2.
\]
At this time, the global metric $K_{\Sigma}$ is far from the lower bound $1.5$. 
Therefore, we consider a panning code $\sigma+$ constructed from a known integer code $\sigma$.
The panning code $\sigma+: \mathbb{N}^{+}\rightarrow \mathbb{B}^{*}$ is constructed as follows.
\begin{equation*}
	\sigma+(m):=\left\{\begin{array}{ll}
		1,                               &\text{if } \, m=1\text{,}\\
		0\sigma(m-1),                    &\text{otherwise,}   \\
	\end{array}\right.   
\end{equation*}
for all $m\in \mathbb{N}$.
If the length of the first codeword of integer code $\sigma$ is $1$, 
then the length of the first codeword of $\sigma+$ is $1$ and the length of the second codeword of $\sigma+$ is $2$.

Before considering which integer code to use to construct the $(2,\infty)$-threshold secret sharing scheme, 
we first prove the following lemma.
\begin{lemma} \label{lem3}
Let the integer code $\sigma$ satisfy $L_{\sigma}(1)=1$, $L_{\sigma}(2)=2$ and
\begin{equation} \label{eq4}
	L_{\sigma}(t) \leq a+b\lfloor \log_{2}(t-1) \rfloor,
\end{equation}
for all $3\leq t\in\mathbb{N}$, where $a$ and $b$ are positive constants.
Then
\begin{equation}\label{eq5}
	\lim\limits_{n\to +\infty}\mathcal{L}(n,\sigma)\leq b,
\end{equation}
and when inequality~\eqref{eq4} takes the equal sign, inequality~\eqref{eq5} also takes the equal sign.
\end{lemma}
\begin{proof}
Let $m=\lfloor \log_{2}n \rfloor$ and $l=n-2^m$.
When the integer $n\geq 3$, we have
\begin{equation}  \label{eq6}
\begin{aligned}
	\sum_{t=1}^{n}L_{\sigma}(t)& \leq 3+\sum_{t=3}^{n}\Big( a+b\lfloor \log_{2}(t-1) \rfloor \Big) \\
	&= 3+a(n-2)+b\sum_{t=2}^{n-1}\lfloor \log_{2}t \rfloor         \\
	&= 3+a(n-2)+b\Big(lm+\sum_{d=1}^{m-1}d\cdot2^{d}\Big) \\
	&= 3+a(n-2)+b\Big[lm+(m-2)2^{m}+2\Big]  \\
	&= 3+a(n-2)+b(nm+2l-2n+2). \\
\end{aligned}
\end{equation}
Then, we obtain
\begin{equation*}  
	\begin{aligned}
		\lim\limits_{n\to +\infty}\mathcal{L}(n,\sigma)&\leq  \lim\limits_{n\to +\infty}\dfrac{3+a(n-2)+b(nm+2l-2n+2)}{nm+2l} \\
		&= b+\lim\limits_{n\to +\infty}\dfrac{3+an-2a+2b-2bn}{nm+2l}       \\
		&= b+\lim\limits_{n\to +\infty}\dfrac{(a-2b)n}{nm+2l}       \\
		&= b. \\
	\end{aligned}
\end{equation*}
From the above calculation process, it is easy to know that when inequality~\eqref{eq4} takes the equal sign, inequality~\eqref{eq5} also takes the equal sign.
\end{proof}

Lemma~\ref{lem3} shows that the integer code $\sigma$ should be chosen 
so that the constant $b$ for the panning code $\sigma+$ in inequality~\eqref{eq4} is as small as possible. 
Because the global metric $K_{\Sigma}\leq 1.5$, it is hoped that the constant $b$ in inequality~\eqref{eq4} takes $1.5$. 
In this case, when $n=2$ and $n$ tends to infinity, $\mathcal{L}(n,\sigma+)$ is less than or equal to $1.5$, and it is reasonable to guess that the global metric at this time is better.

Both $\iota$ code and $\eta$ code satisfy that the length of the first codeword is $1$, 
and the constant $b$ in inequality~\eqref{eq4} can be set to be $1.5$.
Therefore, at the end of this subsection, 
we analyze the global metrics for $(2,\infty)$-threshold secret sharing schemes corresponding to 
$\iota+$ code and $\eta+$ code respectively.

First, analyse the global metrics $K_{I+}$ corresponding to $\iota+$ code. 
The $\iota+$ code satisfy $L_{\iota+}(1)=1$, $L_{\iota+}(2)=2$ and
\begin{equation*} 
	L_{\iota+}(t)=L_{\iota}(t-1)+1 \leq \frac{7}{2}+\frac{3}{2}\lfloor \log_{2}(t-1) \rfloor.
\end{equation*}
for all $3\leq t\in\mathbb{N}$.
When $n=4$, we obtain
\[
\mathcal{L}(4,\iota+)=\dfrac{\sum_{t=1}^{4}L_{\iota+}(t)}{ 4\times2+2\times0 } = 1.625.
\]
When $2\leq n<15$, we can directly verify that $\mathcal{L}(n,\iota+)\leq 1.625$.
When $n\geq 16$ (i.e. $m\geq 4$), $a=\frac{7}{2}$ and $b=\frac{3}{2}$ can be substituted into~\eqref{eq6} to get
\begin{equation*}  
	\sum_{t=1}^{n}L_{\sigma}(t) \leq \frac{3}{2}(nm+2l)+\frac{1}{2}n-1. 
\end{equation*}
Thus, we have
\begin{equation*}  
	\begin{aligned}
		\mathcal{L}(n,\iota+)&\leq \dfrac{1.5(nm+2l)+0.5n-1}{nm+2l} \\
		&= 1.5+\dfrac{0.5n-1}{nm+2l}       \\
		&< 1.5+\dfrac{0.5n}{nm}       \\
		&= 1.5+\dfrac{1}{2m}   \\
		&\leq 1.625 . \\
	\end{aligned}
\end{equation*}
Therefore, the global metric corresponding to the $\iota+$ code is $K_{I+}=1.625$.

Second, analyse the global metrics $K_{H+}$ corresponding to $\eta+$ code.
The $\eta+$ code satisfy $L_{\eta+}(1)=1$, $L_{\eta+}(2)=2$ and
\begin{equation*} 
	L_{\eta+}(t)=L_{\eta}(t-1)+1 \leq 4+\frac{3}{2}\lfloor \log_{2}(t-2) \rfloor.
\end{equation*}
for all $3\leq t\in\mathbb{N}$.
When $n=32$, we obtain
\[
\mathcal{L}(32,\eta+)=\dfrac{\sum_{t=1}^{32}L_{\eta+}(t)}{ 32\times5+2\times0 } = 1.61875.
\]
When $2\leq n<512$, we can directly verify that $\mathcal{L}(n,\eta+)\leq 1.61875$.
When $n\geq 512$ (i.e. $m\geq 9$), we consider the following two cases.
\begin{enumerate}
\item $n$ is a power of $2$:

In this case, $n=2^m$ and $nm+2l=nm$.
When $m\geq 9$, we obtain
\begin{equation*}  
	\begin{aligned}
		\sum_{t=1}^{n}L_{\eta+}(t)& \leq 3+\sum_{t=3}^{n}\Big( 4+\frac{3}{2}\lfloor \log_{2}(t-2) \rfloor \Big) \\
		&  = 4n-5+\frac{3}{2}\sum_{t=2}^{n-2}\lfloor \log_{2}t \rfloor         \\
		&  = 4n-5+\frac{3}{2}\Big(\sum_{d=1}^{m-1}d\cdot2^{d}-(m-1)\Big) \\
		&  = 4n-5+\frac{3}{2}\Big[(m-2)2^{m}-m+3\Big]  \\
		&  = \frac{3}{2}nm+n-\frac{3}{2}m-\frac{1}{2}. \\
	\end{aligned}
\end{equation*}
Thus, we have
\begin{equation*}  
	\begin{aligned}
		\mathcal{L}(n,\eta+)&\leq 1.5+\dfrac{n-1.5m-0.5}{nm}       \\
		&< 1.5+\dfrac{n}{nm}       \\
		&= 1.5+\dfrac{1}{m}    \\
		&< 1.61875 . \\
	\end{aligned}
\end{equation*}
\item $n$ is not a power of $2$:

In this case, $n=2^{m}+l$ and $1\leq l \leq 2^{m}-1$.
When $m\geq 9$, we obtain
\begin{equation*}  
	\begin{aligned}
		\sum_{t=1}^{n}L_{\eta+}(t)& \leq  4n-5+\frac{3}{2}\sum_{t=2}^{n-2}\lfloor \log_{2}t \rfloor         \\
		&  = 4n-5+\frac{3}{2}\Big[(l-1)m+\sum_{d=1}^{m-1}d\cdot2^{d}\Big] \\
		&  = 4n-5+\frac{3}{2}\Big[lm-m+(m-2)2^{m}+2\Big]  \\
		&  = \frac{3}{2}\Big(nm+2l\Big)+n-1.5m-2. \\
	\end{aligned}
\end{equation*}
Thus, we have
	\begin{equation*}  
	\begin{aligned}
		\mathcal{L}(n,\eta+)&\leq 1.5+\dfrac{n-1.5m-2}{nm+2l}       \\
		&< 1.5+\dfrac{n}{nm}       \\
		&= 1.5+\dfrac{1}{m}    \\
		&< 1.61875 . \\
	\end{aligned}
\end{equation*}
\end{enumerate}
In summary, when $m\geq 9$, we obtain $\mathcal{L}(n,\eta+)< 1.61875$.
Therefore, the global metric corresponding to the $\eta+$ code is $K_{H+}=1.61875$.

It can be seen from the calculation that the global metrics $K_{I+}=1.625$ and $K_{H+}=1.61875$ 
are close to the lower bound 1.5, 
so $(2,\infty)$-threshold secret sharing schemes $I+$ and $H+$
corresponding to $\iota+$ code and $\eta+$ code have good overall effects.

\subsection{$\lambda$ code achieves $K_{\Lambda}=1.59375$}

In this subsection, we construct a new integer code, termed $\lambda$ code, 
to achieve the global metric $K_{\Lambda}=1.59375$. 
In the previous subsection, the panning codes constructed using the existing integer codes achieve good global metrics.
At present, the best effect is the global metric $K_{H+}=1.61875$, which is achieved by $\eta+$ code.
The structure of $\lambda$ code is related to $\eta+$ code, and the specific structure is as follows.

Let $\boldsymbol{a}\in\mathbb{B}^{m}$ denote a codeword of length $m$, then $\boldsymbol{a}0$ and $\boldsymbol{a}1$ are two codewords of length $m+1$. 
We call the process from $\boldsymbol{a}$ to $\boldsymbol{a}0$ and $\boldsymbol{a}1$ as \emph{the splitting of $\boldsymbol{a}$}.
If $\boldsymbol{a}$ is a codeword of $\sigma$ code, 
and $\boldsymbol{a}$ is replaced by two codewords after $\boldsymbol{a}$ split, 
then $\sigma$ code is said to be split at $\boldsymbol{a}$, 
and $\sigma$ code after the split is noted as $\sigma[\boldsymbol{a}]$.
Obviously, if $\sigma$ code is a prefix code, then $\sigma[\boldsymbol{a}]$ code is still a prefix code.

$\lambda$ code is essentially the code obtained after $\eta+$ code is splits at $\eta+(5)$, $\eta+(8)$, $\eta+(9)$,
$\eta+(14)$, $\eta+(15)$, $\eta+(16)$ and $\eta+(17)$; that is,
$$\lambda=\eta+[\eta+(5),\eta+(8),\eta+(9),\eta+(14),\eta+(15),\eta+(16),\eta+(17)].$$
Table \ref{tab2} lists the first $24$ codewords of $\eta+$ code and $\lambda$ code.
The underlined part in Table \ref{tab2} is related to the split codeword.
$\sum_{t=1}^{24}L_{\lambda}(t)=174$ can be obtained by simple calculations.
When $n\geq 25$, $\lambda(n)=\eta+(n-7)=0\eta(n-8)$.
Therefore, we obtain
\begin{equation*} 
\begin{aligned}
	L_{\lambda}(n)=L_{\eta}(n-8)+1 \leq 4+\frac{3}{2}\lfloor \log_{2}(n-9) \rfloor,
\end{aligned}	
\end{equation*}
for all $25\leq n\in\mathbb{N}$.

\begin{table}[t]
	\centering
	\caption{The first $24$ codewords of $\eta+$ code and $\lambda$ code}\label{tab2}
	\begin{tabular}{|c|c|c|}
		\hline
		$n$  &   $\eta+$ code  & $\lambda$ code   \\
		\hline
		$1$   &      1        &  1           \\
		$2$   &     01      &   01      \\
		$3$   &    0010     &  0010     \\
		$4$   &    00110    &  00110    \\
		$5$   &  \underline{00111} &  \underline{00111}0  \\
		$6$   &   0001000   &  \underline{00111}1     \\
		$7$   &   0001001   &  0001000       \\
		$8$   &  \underline{0001010}   &  0001001     \\
		$9$   &  \underline{0001011}   & \underline{0001010}0    \\
		$10$  &   00011000   & \underline{0001010}1    \\
		$11$   &  00011001   & \underline{0001011}0    \\
		$12$   &  00011010   & \underline{0001011}1   \\
		$13$   &  00011011  &  00011000    \\
		$14$   &  \underline{00011100}  &  00011001    \\
		$15$   &  \underline{00011101}  &  00011010     \\
		$16$   &  \underline{00011110}  &   00011011   \\
		$17$   &  \underline{00011111}  &  \underline{00011100}0    \\
		$18$   &  000010000  &  \underline{00011100}1    \\
		$19$   &  000010001  &  \underline{00011101}0   \\
		$20$   &  000010010  &  \underline{00011101}1    \\
		$21$   &  000010011  &  \underline{00011110}0    \\
		$22$   &  000010100  &  \underline{00011110}1   \\
		$23$   &  000010101  &  \underline{00011111}0     \\
		$24$   &  000010110  &  \underline{00011111}1   \\
		\hline
	\end{tabular}
\end{table}
Next, analyse the global metrics $K_{\Lambda}$ corresponding to $\lambda$ code. 
When $n=16$, we obtain
\[
\mathcal{L}(16,\lambda)=\dfrac{\sum_{t=1}^{16}L_{\lambda}(t)}{ 16\times4+2\times0 } = 1.59375.
\]
When $2\leq n<2048$, we can directly verify that $\mathcal{L}(n,\lambda)\leq 1.59375$.
When $n\geq 2048$ (i.e. $m\geq 11$), we have
\begin{equation*}  
	\begin{aligned}
		\sum_{t=1}^{n}L_{\lambda}(t)& \leq 174+\sum_{t=25}^{n}L_{\lambda}(t) \\
		&  \leq 174+4(n-24)+\frac{3}{2}\sum_{t=25}^{n}\lfloor \log_{2}(t-9) \rfloor         \\
		&  = 78+4n+\frac{3}{2}\sum_{t=16}^{n-9}\lfloor \log_{2}t \rfloor. \\
	\end{aligned}
\end{equation*}
According to the value of $l=n-2^m$, the following two cases are discussed.
\begin{enumerate}
	\item $l\geq 8$:

	When $m\geq 11$, we obtain
	\begin{equation*}  
		\begin{aligned}
	\sum_{t=16}^{n-9}\lfloor \log_{2}t \rfloor &  =\sum_{d=4}^{m-1}d\cdot2^{d}+ m(l-8)\\
			&  = (m-2)2^{m}-32+ml-8m . \\
		\end{aligned}
	\end{equation*}
	Thus, we have
	\begin{equation*}  
		\begin{aligned}
			\mathcal{L}(n,\lambda)&\leq \dfrac{78+4n+1.5\Big[(m-2)2^{m}-32+ml-8m\Big]}{nm+2l}       \\
			&= 1.5+\dfrac{n-12m+30}{nm+2l}       \\
			&< 1.5+\dfrac{n}{nm+2l}    \\
			&< 1.5+\dfrac{1}{m}    \\
			&< 1.59375. \\
		\end{aligned}
	\end{equation*}
	\item $1\leq l< 8$:
	
	When $m\geq 11$, we obtain
	\begin{equation*}  
	\begin{aligned}
		\sum_{t=16}^{n-9}\lfloor \log_{2}t \rfloor &  =\sum_{d=4}^{m-1}d\cdot2^{d}+ (m-1)(l-8)\\
		&  = (m-2)2^{m}-24+ml-8m-l . \\
	\end{aligned}
\end{equation*}
Thus, we have
\begin{equation*}  
	\begin{aligned}
		\mathcal{L}(n,\lambda)&\leq \dfrac{78+4n+1.5\Big[(m-2)2^{m}-24+ml-8m-l\Big]}{nm+2l}       \\
		&= 1.5+\dfrac{n-12m-1.5l+42}{nm+2l}       \\
		&< 1.5+\dfrac{n}{nm+2l}    \\
		&< 1.59375. \\
	\end{aligned}
\end{equation*}
\end{enumerate}
In summary, when $m\geq 11$, we obtain $\mathcal{L}(n,\lambda)< 1.59375$.
Therefore, the global metric corresponding to the $\lambda$ code is $K_{\Lambda}=1.59375$.
Further, we prove that the range of the global metric $K_{\Sigma}$ for the optimal 
$(2,\infty)$-threshold secret sharing scheme is $1.5\leq K_{\Sigma}\leq1.59375$.

\section{Conclusions}\label{sec_con}
In this paper, we propose a new metric $K_{\Sigma}$ for evolving $2$-threshold secret sharing schemes $\Sigma$ 
and study the range of the global metric $K_{\Sigma}$ for the optimal $(2,\infty)$-threshold secret sharing scheme.
First, we show that the global metric $K_{\Sigma}\geq 1.5$ and use the known universal coding of integers to 
construct schemes with good global metrics.
Second, we construct a new integer code, termed $\lambda$ code, to achieve the global metric $K_{\Lambda}=1.59375$. 
This work shows that the range of the global metric $K_{\Sigma}$ for the optimal 
$(2,\infty)$-threshold secret sharing scheme is $1.5\leq K_{\Sigma}\leq1.59375$.
Furthermore, the reachable lower bound of the sum of share sizes for $(2,n)$-threshold secret sharing schemes is proved.
Several issues remain as follows.
\begin{enumerate}
\item Is it possible to construct an integer code whose global metric is strictly less than $1.59375$?
\item The explicit value of $K_{\Sigma}$ of the optimal $(2,\infty)$-threshold secret sharing scheme is still unknown.
\end{enumerate}
\bibliographystyle{IEEEtran}
\bibliography{IEEEabrv,refs}

\begin{thebibliography}{10}
\providecommand{\url}[1]{#1}
\csname url@samestyle\endcsname
\providecommand{\newblock}{\relax}
\providecommand{\bibinfo}[2]{#2}
\providecommand{\BIBentrySTDinterwordspacing}{\spaceskip=0pt\relax}
\providecommand{\BIBentryALTinterwordstretchfactor}{4}
\providecommand{\BIBentryALTinterwordspacing}{\spaceskip=\fontdimen2\font plus
\BIBentryALTinterwordstretchfactor\fontdimen3\font minus
  \fontdimen4\font\relax}
\providecommand{\BIBforeignlanguage}[2]{{%
\expandafter\ifx\csname l@#1\endcsname\relax
\typeout{** WARNING: IEEEtran.bst: No hyphenation pattern has been}%
\typeout{** loaded for the language `#1'. Using the pattern for}%
\typeout{** the default language instead.}%
\else
\language=\csname l@#1\endcsname
\fi
#2}}
\providecommand{\BIBdecl}{\relax}
\BIBdecl

\bibitem{SS1}
A.~{Shamir}, ``How to share a secret,'' \emph{Communications of the ACM},
  vol.~22, no.~11, pp. 612--613, Nov. 1979.

\bibitem{SS2}
G.~R. {Blakley}, ``Safeguarding cryptographic keys,'' in \emph{1979
  International Workshop on Managing Requirements Knowledge (MARK)}, 1979, pp.
  313--318.

\bibitem{VSS}
M.~K. {Franklin} and M.~K. {Reiter}, ``Verifiable signature sharing,'' in
  \emph{Advances in Cryptology --- EUROCRYPT '95}.\hskip 1em plus 0.5em minus
  0.4em\relax Berlin, Heidelberg: Springer Berlin Heidelberg, 1995, pp. 50--63.

\bibitem{TC1}
Y.~{Desmedt} and Y.~{Frankel}, ``Threshold cryptosystems,'' in \emph{Advances
  in Cryptology --- CRYPTO' 89 Proceedings}.\hskip 1em plus 0.5em minus
  0.4em\relax New York, NY: Springer New York, 1990, pp. 307--315.

\bibitem{TC2}
Y.~{Desmedt} and Y.~{Frankel}, ``Shared generation of authenticators and
  signatures,'' in \emph{Advances in Cryptology --- CRYPTO '91}.\hskip 1em plus
  0.5em minus 0.4em\relax Berlin, Heidelberg: Springer Berlin Heidelberg, 1992,
  pp. 457--469.

\bibitem{EV99}
B.~{Schoenmakers}, ``A simple publicly verifiable secret sharing scheme and its
  application to electronic voting,'' in \emph{Advances in Cryptology ---
  CRYPTO' 99}.\hskip 1em plus 0.5em minus 0.4em\relax Berlin, Heidelbergz:
  Springer Berlin Heidelberg, 1999, pp. 148--164.

\bibitem{GSS}
M.~Ito, A.~Saito, and T.~Nishizeki, ``Secret sharing scheme realizing general
  access structure,'' \emph{Electronics and Communications in Japan (Part III:
  Fundamental Electronic Science)}, vol.~72, no.~9, pp. 56--64, Sept. 1989.

\bibitem{ESS0}
I.~{Komargodski}, M.~{Naor}, and E.~{Yogev}, ``How to share a secret,
  infinitely,'' in \emph{Theory of Cryptography}.\hskip 1em plus 0.5em minus
  0.4em\relax Berlin, Heidelberg: Springer Berlin Heidelberg, 2016, pp.
  485--514.

\bibitem{ESS1}
I.~{Komargodski}, M.~{Naor}, and E.~{Yogev}, ``How to share a secret,
  infinitely,'' \emph{IEEE Transactions on Information Theory}, vol.~64, no.~6,
  pp. 4179--4190, Jun. 2018.

\bibitem{ESS2}
P.~{D'Arco}, R.~{De Prisco}, and A.~{De Santis}, ``On the equivalence of
  2-threshold secret sharing schemes and prefix codes,'' in \emph{10th
  International Symposium, CSS 2018}, 2018, pp. 157--167.

\bibitem{ESS3}
R.~{Okamura} and H.~{Koga}, ``New constructions of an evolving 2-threshold
  scheme based on binary or d-ary prefix codes,'' in \emph{2020 International
  Symposium on Information Theory and Its Applications (ISITA)}, 2020, pp.
  432--436.

\bibitem{RESS}
I.~Komargodski and A.~Paskin-Cherniavsky, ``Evolving secret sharing: Dynamic
  thresholds and robustness,'' in \emph{Theory of Cryptography}.\hskip 1em plus
  0.5em minus 0.4em\relax Cham: Springer International Publishing, 2017, pp.
  379--393.

\bibitem{PESS}
P.~D'Arco, R.~D. Prisco, A.~D. Santis, A.~P. del Pozo, and U.~Vaccaro,
  ``Probabilistic secret sharing,'' in \emph{43rd International Symposium on
  Mathematical Foundations of Computer Science (MFCS 2018)}, ser. Leibniz
  International Proceedings in Informatics (LIPIcs), vol. 117.\hskip 1em plus
  0.5em minus 0.4em\relax Dagstuhl, Germany: Schloss Dagstuhl--Leibniz-Zentrum
  fuer Informatik, 2018, pp. 64:1--64:16.

\bibitem{ERSS1}
A.~Beimel and H.~Othman, ``Evolving ramp secret-sharing schemes,'' in
  \emph{Security and Cryptography for Networks}.\hskip 1em plus 0.5em minus
  0.4em\relax Cham: Springer International Publishing, 2018, pp. 313--332.

\bibitem{ERSS2}
A.~Beimel and H.~Othman, ``Evolving ramp secret sharing with a small gap,'' in
  \emph{Advances in Cryptology -- EUROCRYPT 2020}.\hskip 1em plus 0.5em minus
  0.4em\relax Cham: Springer International Publishing, 2020, pp. 529--555.

\bibitem{yan}
W.~{Yan} and S.-J. {Lin}, ``On the minimum of the expansion factor for
  universal coding of integers,'' \emph{IEEE Transactions on Communications},
  vol.~69, no.~11, pp. 7309--7319, Nov. 2021.

\bibitem{yan2}
W.~Yan and S.-J. Lin, ``A tighter upper bound of the expansion factor for
  universal coding of integers and its code constructions,'' \emph{IEEE
  Transactions on Communications}, 2022, doi:10.1109/TCOMM.2022.3171840.

\bibitem{Elias75}
P.~{Elias}, ``Universal codeword sets and representations of the integers,''
  \emph{IEEE Transactions on Information Theory}, vol.~21, no.~2, pp. 194--203,
  Mar. 1975.

\bibitem{UR}
J.~{Kilian} and N.~{Nisan}, Unpublished result, 1990.

\bibitem{ESS4}
I.~{Cascudo}, R.~{Cramer}, and C.~{Xing}, ``Bounds on the threshold gap in
  secret sharing and its applications,'' \emph{IEEE Transactions on Information
  Theory}, vol.~59, no.~9, pp. 5600--5612, Sep. 2013.

\bibitem{kraft}
L.~G. {Kraft}, ``A device for quantizing, grouping, and coding
  amplitude-modulated pulses,'' Master's thesis, Thesis (M.S.) Massachusetts
  Institute of Technology. Dept. of Electrical Engineering, Cambridge, Mass.,
  1949.

\end{thebibliography}
\end{document}